\documentclass[journal]{IEEEtran}  

\usepackage[T1]{fontenc}
\usepackage[utf8]{inputenc}  

\usepackage{amsmath, amsthm, amssymb, amsfonts, bm, mathtools, dsfont, nccmath}
\theoremstyle{plain}
\newtheorem{Theorem}{Theorem}
\theoremstyle{plain}
\newtheorem{Proposition}{Proposition}

\usepackage{graphicx, epsfig, epstopdf}
\usepackage{caption, subcaption}
\usepackage{float, stfloats}  

\usepackage{algorithm}
\usepackage{algpseudocode}

\usepackage{tabularx}
\usepackage{enumerate}
\usepackage{xspace}
\usepackage{enumitem}

\usepackage[sort,compress]{cite}      
\usepackage[hidelinks]{hyperref}      

\theoremstyle{plain}

\usepackage{balance}

\usepackage{xcolor}

\usepackage[framemethod=tikz]{mdframed}
\definecolor{mycolor}{rgb}{0.122, 0.435, 0.698}
\newmdenv[innerlinewidth=0.5pt, roundcorner=4pt, linecolor=mycolor,
innerleftmargin=6pt, innerrightmargin=6pt,
innertopmargin=6pt, innerbottommargin=6pt]{mybox}

\newcommand{\bmQ}{\mathbf{Q}}

\newcommand{\bmk}{\mathbf{k}}

\newcommand{\bmv}{\mathbf{v}}  
  
  \newcommand{\bmW}{\mathbf{W}}

  \newcommand{\bmr}{\mathbf{r}}

\usepackage[normalem]{ulem}
\usepackage[user]{zref}

\newcounter{revc}
\makeatletter
\zref@newprop{revcontent}{} \zref@addprop{main}{revcontent}
\zref@newprop{revsec}{}     \zref@addprop{main}{revsec}
\zref@newprop{revpage}{}    \zref@addprop{main}{revpage}

\newcommand{\revi}[2]{%
	\zref@setcurrent{revsec}{\thesection}%
	\zref@setcurrent{revpage}{\thepage}%
	\zref@setcurrent{revcontent}{#2}%
	\refstepcounter{revc}%
	\label{#1}%
	\zlabel{#1}%
	\textcolor{blue}{#2}%
}
\newcommand{\revinu}[2]{%
	\zref@setcurrent{revsec}{\thesection}%
	\zref@setcurrent{revcontent}{#2}%
	\refstepcounter{revc}%
	\zlabel{#1}%
	\label{#1}%
	#2
}
\newcommand{\revr}[2]{%
	\zref@setcurrent{revsec}{\thesection}%
	\zref@setcurrent{revcontent}{#2}%
	\refstepcounter{revc}%
	\zlabel{#1}%
	\label{#1} \sout{#2}
}
\makeatother

\expandafter\def\expandafter\quote\expandafter{\quote\onehalfspacing\fontsize{12}{14}\selectfont}

\def\BibTeX{{\rm B\kern-.05em{\sc i\kern-.025em b}\kern-.08em
		T\kern-.1667em\lower.7ex\hbox{E}\kern-.125emX}}

\DeclareUnicodeCharacter{2212}{\textendash}
\DeclareUnicodeCharacter{03A8}{\textendash}

\begin{document}

\title{Secrecy Rate Maximization for 6G Reconfigurable Holographic Surfaces Assisted Systems}  
\author{Chandan Kumar Sheemar,~\IEEEmembership{Member,~IEEE}, Wali Ullah Khan,~\IEEEmembership{Member,~IEEE}, Marco di Renzo,~\IEEEmembership{Fellow,~IEEE}, \\Asad Mahmood,~\IEEEmembership{Member,~IEEE}, Jorge Querol,~\IEEEmembership{Member,~IEEE}, and Symeon Chatzinotas,~\IEEEmembership{Fellow,~IEEE}  
\thanks{Chandan Kumar Sheemar, Wali Ullah Khan, Asad Mahmood and Symeon Chatzinotas are with the SnT department at the University of Luxembourg (email:\{chandankumar.sheemar, waliullah.khan, asad.mahmood, jorge.querol, symeon.chatzinotas\}@uni.lu).  }
\thanks{M. Di Renzo is with Universit\'e Paris-Saclay, CNRS, CentraleSup\'elec, Laboratoire des Signaux et Syst\`emes, 3 Rue Joliot-Curie, 91192 Gif-sur-Yvette, France. (marco.di-renzo@universite-paris-saclay.fr), and with King's College London, Centre for Telecommunications Research -- Department of Engineering, WC2R 2LS London, United Kingdom (marco.di\_renzo@kcl.ac.uk)}
\thanks{The work of M. Di Renzo was supported in part by the European Union through the Horizon Europe project COVER under grant agreement number 101086228, the Horizon Europe project UNITE under grant agreement number 101129618, the Horizon Europe project INSTINCT under grant agreement number 101139161, and the Horizon Europe project TWIN6G under grant agreement number 101182794, as well as by the Agence Nationale de la Recherche (ANR) through the France 2030 project ANR-PEPR Networks of the Future under grant agreement NF-YACARI 22-PEFT-0005, and by the CHIST-ERA project PASSIONATE under grant agreements CHIST-ERA-22-WAI-04 and ANR-23-CHR4-0003-01.} 
}

\maketitle

\maketitle

\begin{abstract}
 Reconfigurable holographic surfaces (RHS) have emerged as a transformative material technology, enabling dynamic control of electromagnetic waves to generate versatile holographic beam patterns. This paper addresses the problem of secrecy rate maximization for an RHS-assisted systems by joint designing the digital beamforming, artificial noise (AN), and the analog holographic beamforming. However, such a problem results to be non-convex and challenging. Therefore, to solve it, a novel alternating optimization algorithm based on the majorization-maximization (MM) framework for RHS-assisted systems is proposed, which rely on surrogate functions to facilitate efficient and reliable optimization. In the proposed approach, digital beamforming design ensures directed signal power toward the legitimate user while minimizing leakage to the unintended receiver. The AN generation method projects noise into the null space of the legitimate user’s channel, aligning it with the unintended receiver’s channel to degrade its signal quality. Finally, the holographic beamforming weights are optimized to refine the wavefronts for enhanced secrecy rate performance
 Simulation results validate the effectiveness of the proposed framework, demonstrating significant improvements in secrecy rate compared to the benchmark method.
\end{abstract}

\IEEEpeerreviewmaketitle
 
\section{Introduction} \label{Intro}

\IEEEPARstart{T}{he} advent of sixth-generation (6G) wireless communications heralds a transformative era characterized by unprecedented data rates, ultra-reliable low-latency communications (URLLC), and massive connectivity for Internet-of-Everything (IoE) applications \cite{10054381,dang2020should}. Beyond conventional connectivity, 6G envisions supporting intelligent environments with integrated sensing, computation, and communication \cite{10812728,sheemar2023full}. This vision necessitates advanced physical-layer technologies to meet the growing demands for spectral efficiency, energy sustainability, and security \cite{10634051,10584518}. As the proliferation of wireless applications in critical sectors such as healthcare, finance, and national security continues to rise, ensuring the confidentiality and integrity of transmitted information becomes a cornerstone of 6G system design \cite{9524814}.

With the emergence of new applications envisioned in 6G, wireless systems are expected to handle increasingly sensitive data while operating in more dynamic and complex environments \cite{10044183,9040264,khan2024reconfigurable}. Furthermore, such applications bring highly heterogeneous requirements in terms of security \cite{10608156}. While some applications prioritize data confidentiality, others may require resilience against service disruption or protection against device spoofing \cite{zaman2023comprehensive}. The openness of the wireless medium further complicates these demands, as signals are inherently vulnerable to interception and interference \cite{10024837}. In such a landscape, ensuring robust security is no longer a secondary consideration but a core necessity for the successful deployment of 6G systems \cite{10268440}. It demands innovative approaches that go beyond traditional cryptographic methods, addressing vulnerabilities at both the physical layer and higher network layers, and providing adaptive, scalable, and energy-efficient solutions that can meet the diverse needs of 6G applications \cite{xu2023beyond,10676973,suomalainen2025cybersecurity}.

Recently, reconfigurable holographic surfaces (RHS) have emerged as a groundbreaking paradigm in wireless communications, offering unprecedented capabilities to enhance the performance and security of wireless systems \cite{9136592,10841801}. An RHS is an array of programmable elements that can dynamically manipulate electromagnetic waves, enabling fine-grained control over the wireless propagation environment \cite{10232975}. By leveraging advanced materials and holographic principles, RHS can steer, focus, and scatter signals with high precision, adapting to varying channel conditions, user locations, and interference patterns \cite{9110848}. Unlike traditional antenna arrays, RHS is highly energy-efficient, operating with minimal power consumption while delivering exceptional control over wavefronts \cite{10746336}.

From a structural standpoint, an RHS consists of an array of sub-wavelength elements \cite{10163760}. These surfaces are typically constructed using low-cost metasurfaces integrated with embedded controllers \cite{10158690}. By applying carefully designed holographic principles, RHS can synthesize complex electromagnetic fields, enabling advanced functionalities such as beamforming, wavefront shaping, and interference suppression \cite{9826717}. This adaptability allows RHS to create custom propagation environments that enhance spectral efficiency and mitigate interference, making them an essential component for 6G communication systems \cite{10263988}. Their ability to optimize the wireless environment in real-time positions RHS as a versatile and cost-effective technology for improving network capacity and reliability.

The potential of RHS to revolutionize physical-layer security is particularly compelling. By dynamically steering signal beams away from potential eavesdroppers and creating nulls in their directions, RHS can significantly degrade the signal-to-noise ratio (SNR) of unauthorized receivers due to the potential of generating pencil-like beams \cite{10003076}. Unlike traditional phased-array systems, RHS-assisted transceivers can rely on passive or minimally powered programmable metasurfaces, which eliminate the need for expensive radio-frequency (RF) components such as phase shifters and other bulky mechanical components \cite{iacovelli2024holographic}. With RHS, a new paradigm of holographic beamforming arises which enables the optimization with only amplitudes instead of phases modulation as in traditional phased array antennas, leading to a cost-effective and energy-efficient solution \cite{deng2021reconfigurable}. This reduction in complexity and reliance on cost-effective materials makes RHS a more economical and scalable solution for large-scale deployments to enable secure communications.

\subsection{State-of-the-Art}
The exploration of RHS for advanced wireless communication systems has led to significant innovations in beamforming, capacity optimization, and energy efficiency. A range of recent studies have examined various aspects of RHS design and application, each offering unique insight into their capabilities and limitations.

One of the foundational studies, \cite{an2023tutorial}, provides a detailed examination of the spatial degrees of freedom (DoF) and ergodic capacity of point-to-point RHS setups. In \cite{ruiz2023degrees}, the authors investigate the DoF, which dictate the number of orthogonal communication modes that can be established for line-of-sight (LoS) holographic communications. In \cite{ruiz2023low}, low complexity designs for reconfigurable intelligent surfaces (RIS)-aided holographic communication systems are presented, in which
the configurations of the holographic transmitter and the RIS are given in a closed-form. In the context of hybrid beamforming, \cite{deng2021reconfigurable} addresses the sum-rate maximization problem by optimizing both digital and analog beamforming components in RHS-assisted systems. The authors propose an iterative method that divides the non-convex optimization into convex subproblems, enabling efficient solutions. While \cite{deng2021reconfigurable} emphasizes hybrid beamforming, \cite{hu2023holographic} presents an alternative approach focused solely on fully analog beamforming. In this study, the authors propose broadcasting a single data stream to all users to maximize the sum-rate, eliminating the need for costly RF chains. However, this method sacrifices flexibility in data multiplexing and interference management, highlighting trade-offs between cost-efficiency and system adaptability. Energy efficiency is another critical area of exploration. The work in \cite{you2022energy} tackles energy efficiency maximization in RHS systems with analog beamforming, addressing both perfect and imperfect channel state information (CSI). This study underlines the practical importance of robust designs for real-world applications, demonstrating solutions that balance efficiency and resilience against CSI inaccuracies. In \cite{sheemar2025holographic_JCAS}, the potential of holographic joint communication and sensing under the Cramer-Rao bounds is investigated. In \cite{sheemar2025joint}, a novel joint holographic beamforming and 3D location optimization for holographic unmanned aerial vehicle (UAV) communications is proposed. In \cite{sheemar2025joint}, the problem of joint holographic beamforming and user scheduling under individual quality-of-service constraints is investigated.

Dynamic metasurface antennas (DMAs) represent another key development in RHS technology. In \cite{shlezinger2019dynamic,shlezinger2021dynamic}, the authors analyze DMAs as a metasurface-based implementation framework for holographic massive MIMO (HMIMO) systems. These works delve into the operational principles of DMAs, discussing their potential in 6G communications and the challenges posed by their broad application scope. The findings in these studies illustrate the advantages of DMAs in achieving dynamic, programmable wireless environments. Further advancing the theoretical foundation of RHS systems, \cite{zhu2024electromagnetic} presents electromagnetic (EM) information theory applied to HMIMO system design. The authors outline an interdisciplinary framework encompassing modeling, analytical methods, and practical applications, setting the stage for innovative EM-domain system designs. Complementing these studies, \cite{wan2021terahertz} explores the application of HMIMO systems for terahertz (THz) communications, identifying the potential of RHS in enabling ultra-high-frequency operations.

Near-field applications of RHS technology also show promise in near-field communications, as explored in \cite{ji2023extra,gan2022near,gong2024holographic,wei2023tri}. In \cite{ji2023extra}, the focus is on spatially constrained antenna apertures with rectangular symmetry, leveraging evanescent waves to enhance spatial degrees of freedom and channel capacity in near-field scenarios. The use of Fourier plane-wave series expansion enables efficient transmission improvements, particularly in confined spatial domains. The integration of RHS into millimeter-wave systems is discussed in \cite{gan2022near}, where the authors investigate advanced beamforming and localization techniques. By optimizing beam patterns and reducing localization errors, the study demonstrates the dual benefits of enhanced communication and positioning accuracy. Similarly, \cite{gong2024holographic} examines generalized EM-domain near-field channel models for point-to-point RHS setups in line-of-sight (LoS) environments. The authors propose two channel models—one emphasizing precision and the other computational simplicity—and derive tight capacity bounds, offering insights into the trade-offs between accuracy and complexity. A unique perspective is introduced in \cite{wei2023tri}, where triple polarization (TP) is utilized for multi-user communication systems with RHS. The study highlights the potential of TP in enhancing system capacity and diversity without increasing the antenna array size. By employing a dyadic Green's function-based TP near-field channel model, the authors propose effective precoding schemes to mitigate cross-polarization and inter-user interference, showcasing the adaptability of RHS for diverse use cases.

The potential of RHS for enhancing wireless communication security has been recently explored in \cite{de2024sum,xu2024reconfigurable}. In \cite{de2024sum}, the authors address the problem of secrecy rate maximization for RHS-assisted transceivers equipped with a single RF chain. They focus on analyzing the sum-secrecy rate and propose a max-min power allocation scheme to optimize system performance under security constraints. However, this work does not incorporate RHS. In \cite{xu2024reconfigurable}, the authors address the problem of secrecy rate maximization for a holographic transceiver. However, the holographic beamforming is designed using a traversal-based method, which results in sub-optimal system performance.

\subsection{Motivation and Main Contributions}
 The exploration of security in RHS-assisted systems remains in its infancy, presenting significant opportunities for further research. While recent works such as \cite{de2024sum} and \cite{xu2024reconfigurable} have highlighted the potential of RHS in enhancing secrecy performance, the scope of these studies remains limited. These gaps underscore the need for innovative approaches that leverage the full potential of RHS technology, such as advanced optimization algorithms, robust designs for practical constraints, and mechanisms to address dynamic adversarial environments. The promising adaptability, energy efficiency, and spatial precision of RHS make it a compelling area for developing new strategies to secure next-generation wireless networks.

 In this work, we propose a comprehensive framework for secrecy rate maximization in hybrid RHS-assisted systems by jointly optimizing the digital beamformers, artificial noise (AN), and analog holographic beamforming. Unlike prior studies in general, we develop a novel alternating optimization design based on the majorization-minimization (MM) framework, a powerful method for handling non-convex problems, which has not been explored in the context of RHS. This approach enables us to systematically address the intricate challenges associated with holographic beamforming, ensuring efficient and reliable solutions.

The MM-based framework operates by iteratively constructing surrogate functions that upper-bound the original non-convex objective function. These surrogate functions are carefully designed to approximate the behaviour of the true objective while being computationally simpler and convex, allowing for efficient optimization at each iteration. By leveraging the surrogate functions, we first design the digital beamforming to direct signal power toward the legitimate user while minimizing leakage to the unintended receiver by solving convex surrogate problems. Subsequently, the AN generation method is designed based on the effective channel response, projecting the noise into the null space of the legitimate user’s channel while aligning it with the unintended receiver’s channel to degrade its signal-to-interference-plus-noise ratio (SINR). Lastly, the holographic beamforming weights are optimized to enable holographic beamforming, further enhancing the overall performance in terms of system security.

Simulation results demonstrate the effectiveness of the proposed framework, showing significant improvements in the secrecy rate compared to the benchmark method. These findings establish the potential of RHS-assisted systems, combined with advanced optimization frameworks like MM, as transformative technology for achieving secure and efficient wireless communications.

\emph{Paper Organization:} The remainder of this paper is organized as follows. Section \ref{section_2} presents the system model and formulates the optimization problem for secrecy rate maximization of a hybrid RHS-assisted system. Section \ref{section_3} develops the MM-based framework and presents a novel algorithmic design to jointly solve the problem. Finally, Section \ref{sec_4} presents the simulation results, and Section \ref{sec_5} concludes the paper.

\emph{Notations:} In this paper, we adopt a consistent set of notations: Scalars are denoted by lowercase or uppercase letters, while vectors and matrices are represented by bold lowercase and bold uppercase letters, respectively. The transpose, Hermitian transpose, and inverse of a matrix $\mathbf{X}$ are denoted by $\mathbf{X}^\mathrm{T}$, $\mathbf{X}^\mathrm{H}$, and $\mathbf{X}^{-1}$, respectively. Sets are indicated by calligraphic letters (e.g., $\mathcal{X}$), and their cardinality is represented by $|\mathcal{X}|$. Finally, $|\cdot|$ denotes the $l_2$-norm.

\section{System Model} \label{section_2}

\begin{figure}
    \centering
    \includegraphics[width=\linewidth]{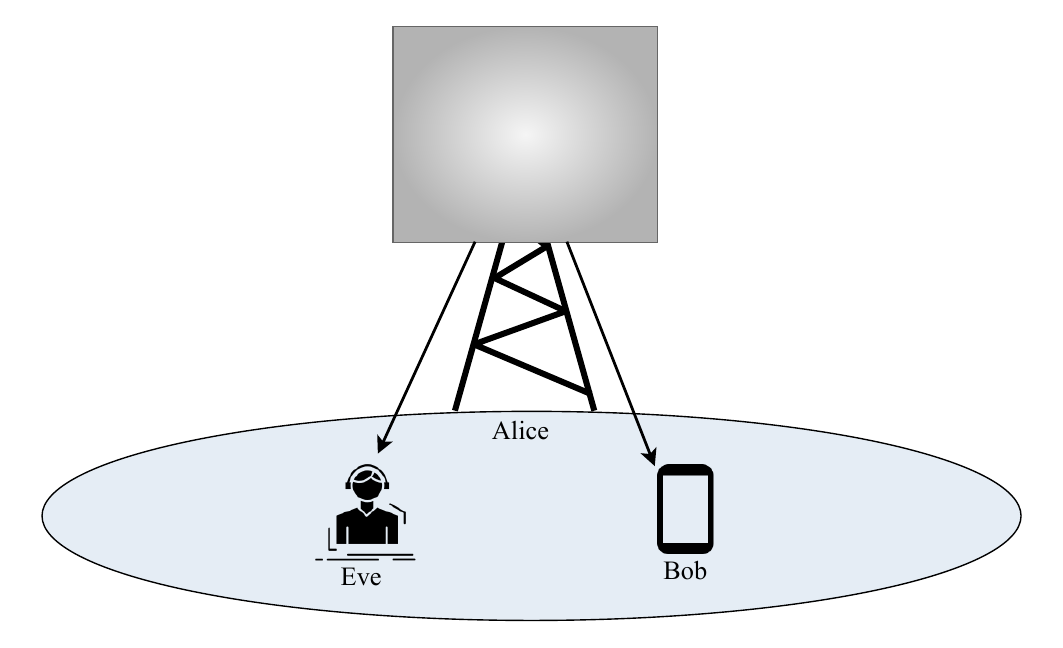}
    \caption{RHS-assisted secure holographic communications.}
    \label{fig1}
\end{figure}
 
\subsection{Scenario Description}
We consider a downlink communication system that consists of one base station (BS) (referred to as Alice), one legitimate user (referred to as Bob), and one intended recipient (referred to as Eve). The BS employs a hybrid beamforming architecture, which consists of two components: a high-dimensional holographic beamformer in the analog domain and a lower-dimensional digital beamforming. The BS is equipped with $R$ RF chains and generates the digitally beamformed signal using the digital beamforming vector $\mathbf{v} \in \mathbb{C}^{R \times 1}$, optimized for Bob.
We assume that each RF chain feeds its corresponding RHS feed, transforming the high-frequency electrical signal into a propagating electromagnetic wave. The digitally processed signal is further multiplied by the analog holographic beamforming matrix $\mathbf{W} \in \mathbb{C}^{M \times R}$, implemented by the RHS. The RHS consists of $M$ reconfigurable elements arranged in a rectangular grid with $\sqrt{M}$ elements along each axis. Each element of the holographic beamforming matrix $\mathbf{W}$ is represented as $\{w_m \cdot e^{-j \bmk_s \cdot \bmr_m}\}$\cite{deng2021reconfigurable},
where $w_m$ denotes the holographic weight of the radiation signal at the $m$-th element to enable holographic beamforming, and $e^{-j \bmk_s \cdot \bmr_m}$ represents the phase shift determined by the reference wave vector $\bmk_s$, and $\bmr_m$ denotes the position of the $m$-th element.

Let $s \sim \mathcal{CN}(\mathbf{0}, 1)$ denote the transmitted data stream to the intended user Bob and let $\mathbf{z} \in \mathbb{C}^{R \times 1}$ denote the artificial noise (AN) vector designed to degrade the reception at Eve. The artificial noise is modelled as $\mathbf{z} \sim \mathcal{CN}(0, \sigma_z^2 \mathbf{I}_R)$. Given the aforementioned notations, the received signal at Bob can be written as
\begin{equation}
y_b = \mathbf{h}_b^H \mathbf{W} \left( \mathbf{v} s + \mathbf{z} \right) + n_b,
\end{equation}
where $\mathbf{h}_b \in \mathbb{C}^{M \times 1}$ is the channel vector between the Alice and Bob, and $n_b \sim \mathcal{CN}(0, \sigma_b^2)$ represents the additive white Gaussian noise (AWGN) at Bob. Similarly, the received signal at Eve can be expressed as
\begin{equation}
y_e = \mathbf{g}_e^H \mathbf{W} \left( \mathbf{v} s + \mathbf{z} \right) + n_e,
\end{equation}
where $\mathbf{h}_e \in \mathbb{C}^{M \times 1}$ is the channel vector between the Alice and Eve, and $n_e \sim \mathcal{CN}(0, \sigma^2)$ represents the AWGN at Eve.

\subsection{Problem Formulation}

 The objective of this work is to maximize the secrecy rate of the Bob, while ensuring that the information leakage to Eve is minimized. The secrecy rate is defined as the difference between the achievable rate at Bob and the rate at which Eve can intercept the information intended for Bob. Mathematically, the secrecy rate is expressed as
\begin{equation}
S = \left[ R_b - R_e \right]^+,
\end{equation}
where $[x]^+ = \max(x, 0)$, $R_b$ is the achievable rate at Bob, and $R_e$ is the achievable rate at Eve. The achievable rate at Bob is given as
\begin{equation}
R_b = \log_2 \left( 1 + \text{SINR}_b \right),
\end{equation}
where the SINR at Bob, denoted as $\text{SINR}_b$, is defined as
\begin{equation}
\text{SINR}_b = \frac{ |\mathbf{h}_b^H \mathbf{W} \mathbf{v}|^2}{| \mathbf{h}_b^H \mathbf{W} \mathbf{z} |^2 + \sigma_b^2}.
\end{equation}
Similarly, the achievable rate at Eve is expressed as
\begin{equation}
R_e = \log_2 \left( 1 + \text{SINR}_e \right),
\end{equation}
where the SINR at Eve is defined as
\begin{equation}
\text{SINR}_e = \frac{\left| \mathbf{g}_e^H \mathbf{W} \mathbf{v} \right|^2}{| \mathbf{g}_e^H \mathbf{W} \mathbf{z} |^2 + \sigma_e^2}.
\end{equation}
The secrecy rate maximization problem is then formulated as
\begin{subequations}
    \begin{equation}
        \max_{\mathbf{v}, \mathbf{W}, \mathbf{z}}   \Big[ \log_2 \left( 1 + \text{SINR}_b \right) - \log_2 \left( 1 + \text{SINR}_e \right) \Big]^+,   
    \end{equation}
    \begin{equation}
       \text{s.t.}\quad \text{Tr}\left(\mathbf{W} \mathbf{v} \mathbf{v}^H \mathbf{W}^H \right) + \text{Tr}\left(\mathbf{z} \mathbf{z}^H\right)  \leq P_t,
    \end{equation}
    \begin{equation}
       \quad 0 \leq w_m \leq 1, \quad \forall m \in \{1, \dots, M\},
    \end{equation}
\end{subequations}
where $P_t$ is Alice's total transmit power budget. Constraint (7b) ensures the total power of the transmitted signal and AN does not exceed the power budget, while constraint (7c) enforces the amplitude of the RHS elements to lie within the range $[0, 1]$ \cite{deng2021reconfigurable}.  

This optimization problem is inherently non-convex, presenting significant challenges for direct optimization. The objective function is a difference of logarithmic terms, each involving the SINR at the Bob and the Eve. These SINR expressions are non-linear and coupled with respect to the optimization variables: the digital beamforming vector $\mathbf{v}$, the holographic beamforming matrix $\mathbf{W}$, and the AN vector $\mathbf{z}$. This coupling introduces complex interdependencies, making the problem non-convex and difficult to solve directly. Moreover, the interplay between the variables introduces a trade-off between enhancing the SINR at the Bob and degrading it at the Eve.

\section{MM-Based Solution} \label{section_3}
To address the challenges posed by the non-convex nature of the secrecy rate maximization problem for holographic beamforming with RHS, we propose an alternating optimization approach that decomposes the joint optimization problem into manageable subproblems. Each subproblem is then solved iteratively, ensuring steady improvement toward the overall objective. Specifically, we develop a novel MM-based framework to achieve this task \cite{7547360}. It constructs surrogate functions that approximate the non-convex objective with simpler convex functions, allowing efficient and reliable optimization at each step. By iteratively updating each variable, the proposed method effectively navigates the complex optimization landscape, achieving robust and secure communication in an RHS-assisted system.

 \subsection{Digital Beamforming Optimization}

To solve the problem, we first consider optimizing the digital beamforming matrix $\mathbf{v}$, assuming that the holographic beamforming matrix $\mathbf{W}$ and the AN vector $\mathbf{z}$ are fixed. The subproblem with respect to the $\bmv$ can be stated as

\begin{subequations}
    \begin{equation}
\max_{\mathbf{v}}  \Big[ \log_2\left( 1 + \text{SINR}_b \right) - \log_2\left( 1 + \text{SINR}_{e} \right) \Big]
\end{equation}
\begin{equation}
 \text{s.t.}\quad \text{Tr}\left(\mathbf{W} \mathbf{v} \mathbf{v}^H \mathbf{W}^H \right) \leq P_t - \text{Tr}\left(\mathbf{z} \mathbf{z}^H\right).
\end{equation}
\end{subequations}
This problem is non-convex due to the coupling of the digital beamformer $\bmv$ in the numerator and denominator of the SINR at Bob and Eve. To solve it, we first reformulate it to feed it to the MM framework.

To solve the secrecy rate maximization problem, the problem needs to be reformulated with surrogate functions that simplify the inherently non-convex optimization. By employing the MM-based framework, we majorize each logarithmic term in the objective function with a carefully constructed convex surrogate function. These surrogate functions are designed to closely approximate the behaviour of the original objective while being computationally tractable. 

To proceed, we first derive an important result to build the surrogate functions for secrecy rate maximization in an RHS-assisted system.
\begin{Theorem}
     Let $ f(\mathbf{v}) = \log_2\left(1 + \frac{\mathbf{v}^H \mathbf{W}^H \mathbf{H} \mathbf{W} \mathbf{v}}{b}\right) $, where $\mathbf{W} \in \mathbb{C}^{M \times F}$ is the holographic beamformer, $\mathbf{H} \succeq 0$ (positive semidefinite) represents the channel, and $b > 0$ is a constant. Then $f(\mathbf{v})$ can be majorized as
\begin{equation}
f(\mathbf{v})\hspace{-1mm} \leq f(\mathbf{v}^{(t)}) + \frac{\mathbf{v}^{(t)H} \mathbf{W}^H \mathbf{H} \mathbf{W} \mathbf{v}^{(t)}}{\ln(2) \left(b + \mathbf{v}^{(t)H} \mathbf{W}^H \mathbf{H} \mathbf{W} \mathbf{v}^{(t)}\right)} \mathbf{v}^H \mathbf{W}^H \mathbf{H} \mathbf{W} \mathbf{v},
\end{equation}
where $\mathbf{v}^{(t)}$ is the value of $\mathbf{v}$ at the $t$-th iteration.
\end{Theorem}
\begin{proof}
    The proof is provided in Appendix 1.
\end{proof}

Using this result, the achievable rate terms for Bob and Eve can be approximated based on the following propositions.

\begin{Proposition}
    The rate for Bob \(\log_2(1 + \text{SINR}_b)\) can be majorized as
\begin{equation}
\log_2\left(1 + \frac{\mathbf{v}^H \mathbf{W}^H \mathbf{h}_b \mathbf{h}_b^H \mathbf{W} \mathbf{v}}{|\mathbf{h}_b^H \mathbf{W} \mathbf{z}|^2 + \sigma_b^2}\right) \leq \text{c}_b^{(t)} + \frac{\mathbf{v}^H \mathbf{M}_b^{(t)} \mathbf{v}}{\ln(2)},
\end{equation}
where \(\text{c}_b^{(t)}\) and \(\mathbf{M}_b^{(t)}\) are given as
\begin{equation}
\text{c}_b^{(t)} = \log_2\left(1 + \frac{\mathbf{v}^{(t)H} \mathbf{W}^H \mathbf{h}_b \mathbf{h}_b^H \mathbf{W} \mathbf{v}^{(t)}}{|\mathbf{h}_b^H \mathbf{W} \mathbf{z}|^2 + \sigma_b^2}\right),
\end{equation}
\begin{equation}
\mathbf{M}_b^{(t)} = \frac{\mathbf{W}^H \mathbf{h}_b \mathbf{h}_b^H \mathbf{W}}{\left(|\mathbf{h}_b^H \mathbf{W} \mathbf{z}|^2 + \sigma_b^2\right) \left(1 + \frac{\mathbf{v}^{(t)H} \mathbf{W}^H \mathbf{h}_b \mathbf{h}_b^H \mathbf{W} \mathbf{v}^{(t)}}{|\mathbf{h}_b^H \mathbf{W} \mathbf{z}|^2 + \sigma_b^2}\right)}.
\end{equation}
\end{Proposition}

\begin{proof}
    The proof follows directly from Theorem $1$.
\end{proof}
\begin{Proposition}
   The rate for Eve \(\log_2(1 + \text{SINR}_e)\) can be majorized as
\begin{equation}
\log_2\left(1 + \frac{\mathbf{v}^H \mathbf{W}^H \mathbf{g}_e \mathbf{g}_e^H \mathbf{W} \mathbf{v}}{|\mathbf{g}_e^H \mathbf{W} \mathbf{z}|^2 + \sigma_e^2}\right) \leq \text{c}_e^{(t)} + \frac{\mathbf{v}^H \mathbf{M}_e^{(t)} \mathbf{v}}{\ln(2)},
\end{equation}
where \(\text{c}_e^{(t)}\) and \(\mathbf{M}_e^{(t)}\) are given as
\begin{equation}
\text{c}_e^{(t)} = \log_2\left(1 + \frac{\mathbf{v}^{(t)H} \mathbf{W}^H \mathbf{g}_e \mathbf{g}_e^H \mathbf{W} \mathbf{v}^{(t)}}{|\mathbf{g}_e^H \mathbf{W} \mathbf{z}|^2 + \sigma_e^2}\right),
\end{equation}
\begin{equation}
\mathbf{M}_e^{(t)} = \frac{\mathbf{W}^H \mathbf{g}_e \mathbf{g}_e^H \mathbf{W}}{\left(|\mathbf{g}_e^H \mathbf{W} \mathbf{z}|^2 + \sigma_e^2\right) \left(1 + \frac{\mathbf{v}^{(t)H} \mathbf{W}^H \mathbf{g}_e \mathbf{g}_e^H \mathbf{W} \mathbf{v}^{(t)}}{|\mathbf{g}_e^H \mathbf{W} \mathbf{z}|^2 + \sigma_e^2}\right)}.
\end{equation}
\end{Proposition}
\begin{proof}
    The proof follows directly from Theorem $1$.
\end{proof}
Using the surrogate functions derived above, we can restate the original optimization problem for digital beamforming as
\begin{subequations} \label{restated_SR}
    \begin{equation}
\max_{\mathbf{v}} \frac{\mathbf{v}^H \bmQ_v \mathbf{v}}{\ln(2)}
\end{equation}
\begin{equation}
 \text{s.t.}\quad \text{Tr}\left(\mathbf{W} \mathbf{v} \mathbf{v}^H \mathbf{W}^H \right) \leq P_t - \text{Tr}\left(\mathbf{z} \mathbf{z}^H\right),
\end{equation}
\end{subequations}
where $\bmQ_v = (\mathbf{M}_b^{(t)} - \mathbf{M}_e^{(t)})$. The reformulated problem utilizes surrogate functions to restate the original secrecy rate maximization task in a more tractable form. By replacing the non-convex logarithmic terms with convex surrogates, the problem becomes a quadratic optimization. The modular structure allows for the optimization over \(\mathbf{v}\) while treating \(\mathbf{W}\) and \(\mathbf{z}\) as fixed parameters in each iteration.

\begin{Theorem}
The optimal unconstrained digital beamformer $\mathbf{v}$ for the secrecy rate maximization problem is obtained by solving the generalized eigenvalue problem
\begin{equation}
\mathbf{Q}_v^{(t)} \mathbf{v} = \lambda \mathbf{R} \mathbf{v},
\end{equation}
where $\mathbf{Q}_v^{(t)} = \mathbf{M}_b^{(t)} - \mathbf{M}_e^{(t)}$ is the effective secrecy rate optimization matrix, and $\mathbf{R} = \mathbf{W}^H \mathbf{W}$ is the effective power weight matrix. The eigenvector $\mathbf{u}_{\max}$ corresponding to the largest eigenvalue $\lambda_{\max}$ is the optimal $\mathbf{v}$:
\begin{equation}
\mathbf{v} = \mathbf{u}_{\max}(\mathbf{Q}_v^{(t)}, \mathbf{R}),
\end{equation}
where $\mathbf{u}_{\max}(\cdot)$ denotes the eigenvector associated with the largest eigenvalue.
\end{Theorem}
\begin{proof}
To establish the theorem, we begin by analyzing the optimization problem for the digital beamforming vector $\mathbf{v}$. Specifically, we consider the reformulated objective function, given by

\begin{equation}
\max_{\mathbf{v}} \quad \frac{\mathbf{v}^H \mathbf{Q}_v^{(t)} \mathbf{v}}{\mathbf{v}^H \mathbf{R} \mathbf{v}},
\end{equation}

where $\mathbf{Q}_v^{(t)} = \mathbf{M}_a^{(t)} - \mathbf{M}_b^{(t)}$ is a Hermitian matrix (since it is constructed as the difference of two Hermitian matrices), and $\mathbf{R} = \mathbf{W}^H \mathbf{W}$ is a positive semi-definite matrix representing the power constraint.

To derive the optimal solution, we recognize that this problem falls into the class of generalized Rayleigh quotient maximization problems, which are well-known to be solved via the generalized eigenvalue decomposition (see Theorem 8.1.1  in \cite{golub1996cf}). Specifically, the standard result states that the optimal vector $\mathbf{v}$ is given by the eigenvector corresponding to the largest generalized eigenvalue of the matrix pair $(\mathbf{Q}_v^{(t)}, \mathbf{R})$, leading to the eigenvalue equation

\begin{equation}
\mathbf{Q}_v^{(t)} \mathbf{v} = \lambda \mathbf{R} \mathbf{v}.
\end{equation}

Since maximizing the quadratic form $\mathbf{v}^H \mathbf{Q}v^{(t)} \mathbf{v}$ under the constraint $\mathbf{v}^H \mathbf{R} \mathbf{v} = 1$, which ensures that we are selecting an optimal direction for $\mathbf{v}$ rather than letting its magnitude grow arbitrarily, is equivalent to selecting the eigenvector $\mathbf{u}{\max}$ corresponding to the largest eigenvalue $\lambda_{\max}$, we conclude that

\begin{equation}
\mathbf{v}^{*} = \mathbf{u}_{\max},
\end{equation}

where $\mathbf{v}^{*}$ is the optimal solution to the problem. This result aligns with classical results in quadratic optimization with quadratic constraints and generalized eigenvalue problems, thus the Theorem follows.  
\end{proof}

Next, we address the power constraint.
To enforce this, the beamforming vector is scaled by a factor $\beta$ such that the equality is satisfied with available power. The scaling factor is computed as
\begin{equation}
\beta = \min\left(1, \sqrt{\frac{P_{\text{av}}}{\text{Tr}\left(\mathbf{W} \mathbf{v} \mathbf{v}^H \mathbf{W}^H\right)}}\right),
\end{equation}
where $P_{\text{av}} = P_t - \|\mathbf{z}\|^2$ represents the remaining available power after allocating resources to the artificial noise, computed at the previous iteration. The scaled beamforming vector satisfying the remaining power constraint with equality can be obtained as 
\begin{equation} \label{ref_optimal_digital}
\mathbf{v} \gets \beta \mathbf{v}.
\end{equation}
This ensures that the power constraint is met with equality for the available to maximize the performance.

\subsection{AN Optimization}

In this section, we derive the optimal AN vector $\mathbf{z}$, which is critical for minimizing the SINR at Eve while maximizing the performance for Bob. The optimization problem with respect to $\mathbf{z}$ for secrecy rate maximization can be formulated as 
\begin{subequations}
\begin{equation}
\max_{\mathbf{z}} \left[ \log_2\left( 1 + \text{SINR}_b \right) - \log_2\left( 1 + \text{SINR}_{e} \right) \right],
\end{equation}
\begin{equation}
 \text{Tr}\left(\mathbf{z} \mathbf{z}^H\right)  \leq P_t -  \text{Tr}\left(\mathbf{W} \mathbf{v} \mathbf{v}^H \mathbf{W}^H \right),
\end{equation}
\end{subequations}



To solve this problem, we must ensure that the AN does not interfere with Bob. To this end, we first project $\mathbf{z}$ onto the null space of the effective channel of Bob. Note that the effective channel vector for Bob is given by
\begin{equation}
\mathbf{H}_{b} = \bmW^H \mathbf{h}_b.
\end{equation}
Let $\mathbf{N}_{b}$ denote the null space of $\mathbf{H}_{b}^H$, such that $\mathbf{H}_{b}^H \mathbf{N}_{b} = \mathbf{0}$. The optimal AN vector can them be computed as  
\begin{equation}
\mathbf{z} = \mathbf{N}_{b} \mathbf{z}_0,
\end{equation}
where $\mathbf{z}_0$ is an auxiliary vector to be optimized within the null space. This projection ensures that $\mathbf{z}$ lies in a subspace orthogonal to the effective channel of Alice, thereby avoiding interference with her signal.

The power allocated to digital beamforming and AN must satisfy the total power constraint
\begin{equation}
\text{Tr}\left( \bmW \mathbf{v} \mathbf{v}^H \bmW^H \right) + \text{Tr}\left( \mathbf{z} \mathbf{z}^H\right)  \leq P_t.
\end{equation}

To allocate power for the AN, the current power used by the digital beamforming vector $\mathbf{v}$ is first computed as
\begin{equation}
P_{\text{cur}} = \text{Tr}\left(\mathbf{v}^H \bmW^H \bmW \mathbf{v}\right),
\end{equation}
which dictates the residual power that can be used for AN, given as
\begin{equation}
P_{\text{av}} = P_t - P_{\text{cur}}.
\end{equation}

To degrade the SINR at Eve, the auxiliary vector $\mathbf{z}_0$ must be designed. Note that the covariance matrix of the interference caused by the AN at Eve is given by
\begin{equation}
\mathbf{G}_{e} = \bmW^H \mathbf{g}_e \mathbf{g}_e^H \bmW,
\end{equation}
where $\mathbf{g}_e$ is the channel vector from the RHS to Eve. The vector $\mathbf{z}_0$ can be aligned with the principal eigenvector $\mathbf{u}_{\max}$ of $\mathbf{G}_{e}$, corresponding to its largest eigenvalue
\begin{equation}
\mathbf{z}_0 = \sqrt{P_{\text{av}}} \frac{\mathbf{u}_{\max}}{|\mathbf{u}_{\max}|}.
\end{equation}
and the optimal noise is given as
\begin{equation} \label{noise_optimal}
\mathbf{z} = \mathbf{N}_{b} \mathbf{z}_0.
\end{equation}
This approach guarantees that the total power allocated to digital beamforming and AN respects the inequality constraint while effectively degrading Bob's SINR.  

\subsection{Holographic Beamforming}

Note that based on the considered holographic architecture, we can decompose the holographic beamformer into two parts as $ \bmW = \text{diag}([w_1, w_2, \dots, w_M]) \mathbf{\Phi} $ \cite{
zhang2022holographic}, where the diagonal elements $ w_m $ represent the adjustable weight of the holographic beamformer for the reconfigurable element $0\leq m \leq M$, and $ \mathbf{\Phi} $ contains fixed phase components, which are determined by the reference wave propagating within the waveguide, i.e., with elements of the form $e^{-j \bmk_s \cdot \bmr_{m}^k}$. Each $ w_m $ is constrained to lie in the range $ [0, 1] $. Based on such observation, we can formulate the optimization problem with respect to the holographic beamformer as

 \begin{subequations} \label{holo_problem}
    \begin{equation}
    \begin{aligned}
        \max_{\mathbf{w}}   & \left[ \log_2\left(1 + \frac{| \mathbf{h}_a^H \text{diag}(\mathbf{w}) \mathbf{\Phi} \mathbf{v} |^2}{| \mathbf{h}_a^H \text{diag}(\mathbf{w}) \mathbf{\Phi} \mathbf{z} |^2 + \sigma_a^2}\right) \right. \\
        & \quad \left. - \log_2\left(1 + \frac{| \mathbf{g}_b^H \text{diag}(\mathbf{w}) \mathbf{\Phi} \mathbf{v} |^2}{| \mathbf{g}_b^H \text{diag}(\mathbf{w}) \mathbf{\Phi} \mathbf{z} |^2 + \sigma_b^2}\right) \right],
    \end{aligned}
    \end{equation}
    \begin{equation}
    \textbf{s.t:} \quad
    0 \leq w_m \leq 1, \quad \forall m \in \{1, \dots, M\}.
    \end{equation}
\end{subequations}
In this section, we address this optimization problem by leveraging the MM-based framework developed above. Recall that such a framework simplifies the optimization by replacing the non-convex terms with surrogate functions. For a concave function $\log_2(1 + x)$, the following upper bound holds \cite{boyd2004convex}
\begin{equation}
\log_2(1 + x) \leq \log_2(1 + x^{(t)}) + \frac{x - x^{(t)}}{\ln(2)(1 + x^{(t)})},
\end{equation}
where $x^{(t)}$ is the value of $x$ at the current iteration.

Consider the numerator of Bob's SINR, and the power contribution is expressed as 
 $| \mathbf{h}_b^H \text{diag}(\mathbf{w}) \mathbf{\Phi} \mathbf{v} |^2$. Expanding this expression, the term \(\mathbf{h}_b^H \text{diag}(\mathbf{w}) \mathbf{\Phi} \mathbf{v}\) can be rewritten in matrix form as \(\mathbf{w}^H \text{diag}(\mathbf{\Phi} \mathbf{v}) \mathbf{h}_b\). Taking the norm squared of this term results in a quadratic form
\begin{equation}
| \mathbf{h}_b^H \text{diag}(\mathbf{w}) \mathbf{\Phi} \mathbf{v} |^2 = \mathbf{w}^H \left[\text{diag}(\mathbf{\Phi} \mathbf{v})^H \mathbf{h}_b \mathbf{h}_b^H \text{diag}(\mathbf{\Phi} \mathbf{v})\right] \mathbf{w}.
\end{equation}
Given this property, we define matrix \(\mathbf{M}_b = \text{diag}(\mathbf{\Phi} \mathbf{v})^H \mathbf{h}_b \mathbf{h}_b^H \text{diag}(\mathbf{\Phi} \mathbf{v})\), and the  numerator simplifies to 
\begin{equation}
| \mathbf{h}_b^H \text{diag}(\mathbf{w}) \mathbf{\Phi} \mathbf{v} |^2=\mathbf{w}^T \mathbf{M}_b \mathbf{w}.
\end{equation}

Similarly for Eve, by using the same properties, the numerator can be represented as
\begin{equation}
| \mathbf{g}_e^H \text{diag}(\mathbf{w}) \mathbf{\Phi} \mathbf{v} |^2 = \mathbf{w}^T \mathbf{M}_e \mathbf{w},
\end{equation}
where
\begin{equation}
\mathbf{M}_e = \text{diag}(\mathbf{\Phi} \mathbf{v})^H \mathbf{g}_e \mathbf{g}_e^H \text{diag}(\mathbf{\Phi} \mathbf{v}).
\end{equation}

By applying the same properties and expressing the denominator similarly, we can construct the surrogate functions for Bob and Eve to maximize the secrecy rate based on the result derived in Theorem 1, leading to the expressions in \eqref{surr_bob} and \eqref{surr_eve}, which are shown at the top of the next page.

\begin{figure*}
\begin{equation} \label{surr_bob}
\mathbf{M}_b^{(t)} = \frac{\mathbf{M}_b}{
\left(
\mathbf{w}^{T(t)} 
\left[
\text{diag}(\mathbf{\Phi} \mathbf{z})^H \mathbf{g}_b \mathbf{g}_b^H \text{diag}(\mathbf{\Phi} \mathbf{z})
\right] 
\mathbf{w}^{(t)} + \sigma_b^2
\right)
\left(
1 + 
\frac{\mathbf{w}^{T(t)} \mathbf{M}_b \mathbf{w}^{(t)}}{
\mathbf{w}^{T(t)} 
\left[
\text{diag}(\mathbf{\Phi} \mathbf{z})^H \mathbf{g}_b \mathbf{g}_b^H \text{diag}(\mathbf{\Phi} \mathbf{z})
\right] 
\mathbf{w}^{(t)} + \sigma_b^2
}
\right)
}.
\end{equation}
 \begin{equation}\label{surr_eve}
     \mathbf{M}_e^{(t)} = \frac{\mathbf{M}_e}{
\left(
\mathbf{w}^{T(t)} 
\left[
\text{diag}(\mathbf{\Phi} \mathbf{z})^H \mathbf{h}_a \mathbf{h}_a^H \text{diag}(\mathbf{\Phi} \mathbf{z})
\right] 
\mathbf{w}^{(t)} + \sigma_a^2
\right)
\left(
1 + 
\frac{\mathbf{w}^{T(t)} \mathbf{M}_e \mathbf{w}^{(t)}}{
\mathbf{w}^{T(t)} 
\left[
\text{diag}(\mathbf{\Phi} \mathbf{z})^H \mathbf{h}_a \mathbf{h}_a^H \text{diag}(\mathbf{\Phi} \mathbf{z})
\right] 
\mathbf{w}^{(t)} + \sigma_a^2
}
\right)
}. 
 \end{equation}
 \hrulefill
\end{figure*}

Given such expressions, let $\mathbf{Q}_w^{(t)}$ be defined as
\begin{equation}
\mathbf{Q}_w^{(t)} = \mathbf{M}_b^{(t)} - \mathbf{M}_e^{(t)}.
\end{equation}

Therefore, we can restate the optimization problem 
\eqref{holo_problem} as 
 \begin{subequations} \label{holo_problem_restated}
    \begin{equation}
    \begin{aligned}
        \max_{\mathbf{w}} \quad \mathbf{w}^T \bmQ_w \mathbf{w} ,
    \end{aligned}
    \end{equation}
    \begin{equation}
    \textbf{s.t:} \quad
    0 \leq w_m \leq 1, \quad \forall m \in \{1, \dots, M\}.
    \end{equation}
\end{subequations}

This reformulation transforms the original problem into an iterative optimization task over the quadratic matrix \(\mathbf{Q}^{(t)}\), enabling efficient updates of the holographic weights \(\mathbf{w}\).  To solve \eqref{holo_problem_restated}, we first define the objective function as 
\begin{equation}
f(\mathbf{w}) = \mathbf{w}^T \mathbf{Q}^{(t)} \mathbf{w}.
\end{equation}
and to find the optimal direction for the holographic weights, we take its derivative with respect to $\mathbf{w}$, which leads to the gradient 
\begin{equation}
\nabla_{\mathbf{w}} f(\mathbf{w}) = 2 \mathbf{Q}^{(t)} \mathbf{w}.
\end{equation}
This gradient provides the direction for updating $\mathbf{w}$ to maximize the objective function, indicating the direction in which $f(\mathbf{w})$ increases most rapidly. At each iteration $t$, the vector $\mathbf{w}$ is updated using the gradient ascent method with the following update rule
\begin{equation}
\mathbf{w}^{(t+1)} = \mathbf{w}^{(t)} + \eta \nabla_{\mathbf{w}} f(\mathbf{w}^{(t)}),
\end{equation}
where $\eta > 0$ is the learning rate, determining the step size along the direction of the gradient. To ensure feasibility, the updated vector $\mathbf{w}^{(t+1)}$ is projected onto the box-constrained domain
\begin{equation}
w_m^{(t+1)} = \min(1, \max(0, w_m^{(t+1)})), \quad \forall m \in \{1, \dots, M\}.
\end{equation}
This projection guarantees that the updated elements of $\mathbf{w}$ remain within the range $[0, 1]$, satisfying the physical constraints of the holographic beamformer. The optimize the holographic beamformer, the iterative steps are formally stated in Algorithm $1$.

 \begin{algorithm}[t]
\caption{MM-based holographic beamforming optimization}
\label{alg:gradient_beamforming}
\begin{algorithmic}[1]
\State \textbf{Input:} Initial holographic amplitudes vector $\mathbf{w}^{(0)} \in [0, 1]^M$, learning rate $\eta > 0$, convergence tolerance $\epsilon$, maximum iterations $T_{\max}$.
\State \textbf{Initialize:} Set iteration $t = 0$.
\Repeat
    \State Compute the gradient:
    \begin{equation}
    \nabla_{\mathbf{w}}^{(t)} = 2 \mathbf{Q}^{(t)} \mathbf{w}^{(t)}.
    \end{equation}
    \State Update $\mathbf{w}$ using the gradient ascent:
    \begin{equation}
    \mathbf{w}^{(t+1)} = \mathbf{w}^{(t)} + \eta \nabla_{\mathbf{w}}^{(t)}.
    \end{equation}
    \State Project $\mathbf{w}^{(t+1)}$ onto the feasible region:
    \begin{equation}
    w_m^{(t+1)} = \min(1, \max(0, w_m^{(t+1)})), \quad \forall m \in \{1, \dots, M\}.
    \end{equation}
    \State Set $t \gets t + 1$.
\Until{$\|\mathbf{w}^{(t)} - \mathbf{w}^{(t-1)}\|_2 < \epsilon$ or $t \geq T_{\max}$}
\State \textbf{Output:}  Holographic beamformer $\mathbf{W}^* =  \text{diag}(\textbf{w}^*) \Phi$.
\end{algorithmic}
\end{algorithm}

The joint optimization algorithm, which jointly optimizes the digital beamformer, the artificial noise vector and the analog holographic beamformer is given in Algorithm \ref{alg:joint_optimization}.

\begin{algorithm}[t]
\caption{MM-based secrecy rate maximization for an RHS-assisted system}
\label{alg:joint_optimization}
\begin{algorithmic}[1]
\State \textbf{Input:} Initial digital beamforming matrix $\mathbf{v}^{(0)}$, initial holographic amplitudes vector $\mathbf{W}^{(0)}$, initial artificial noise vector $\mathbf{z}^{(0)}$, convergence tolerance $\epsilon$, and maximum iterations $T_{\max}$.
\State \textbf{Initialize:} Set iteration $t = 0$.
\Repeat
    \State \textbf{Step 1: Digital Beamforming Optimization.}
    \State Compute the optimal unconstrained digital beamformer $\bmv$.
    \State Scale it as \eqref{ref_optimal_digital}.
    \State \textbf{Step 2: Artificial Noise Optimization.}
    \State Solve the artificial noise optimization problem for $\mathbf{z}$ with fixed $\mathbf{v}^{(t+1)}$ and $\mathbf{w}^{(t)}$.
    \State Update the artificial noise vector \eqref{noise_optimal}.
    \State \textbf{Step 3: Holographic Beamforming Optimization.}
    \State Solve the holographic beamforming problem for $\mathbf{W}$ by using Algorithm 1, with fixed $\mathbf{v}^{(t+1)}$ and $\mathbf{z}^{(t+1)}$.
    \State Update the holographic beamformer $\mathbf{W}^{(t+1)}$.
    \State Set $t \gets t + 1$.
    \Until{The objective function converges or $t \geq T_{\max}$}
\State \textbf{Output:}  $\mathbf{v}$, $\mathbf{W}$, and $\mathbf{z}$.
\end{algorithmic}
\end{algorithm}

 \subsection{Convergence Analysis}

The convergence of the proposed joint optimization algorithm is established through the theoretical principles of the MM-based framework and gradient ascent-based alternating optimization. The secrecy rate maximization problem is tackled by iteratively optimizing the digital beamforming vector $\mathbf{v}$, the holographic beamforming vector $\mathbf{w}$, and the artificial noise vector  $\mathbf{z}$. 

The MM framework constructs a surrogate function $ g(\mathbf{x} \mid \mathbf{x}^{(t)}) $ that approximates the original non-convex objective function $ f(\mathbf{x}) $ at each iteration $ t $. This surrogate function satisfies two critical properties:

\begin{enumerate}
    \item \emph{Majoration Property}:  
    \begin{equation}
        f(\mathbf{x}) \leq g(\mathbf{x} \mid \mathbf{x}^{(t)}), \quad \forall \mathbf{x}
    \end{equation}
    ensuring that $ g(\mathbf{x} \mid \mathbf{x}^{(t)}) $ is an upper bound on $ f(\mathbf{x}) $.  

    \item \emph{Tangency Property}:  
    \begin{equation}
        f(\mathbf{x}^{(t)}) = g(\mathbf{x}^{(t)} \mid \mathbf{x}^{(t)})
    \end{equation}
    ensuring that the surrogate function matches the original function at the current solution $\mathbf{x}^{(t)}$.  
\end{enumerate}

These properties guarantee that optimizing the surrogate function $ g(\mathbf{x} \mid \mathbf{x}^{(t)}) $ leads to a monotonic improvement in the objective function at each iteration:
\begin{equation}
    f(\mathbf{x}^{(t+1)}) \geq f(\mathbf{x}^{(t)}).
\end{equation}

The proposed optimization also follows a gradient ascent approach to optimize the holographic beamformer, while keeping the digital beamformer and the others fixed. The update rule follows:
\begin{equation} \label{grad_update}
    \mathbf{x}^{(t+1)} = \mathbf{x}^{(t)} + \eta \nabla_{\mathbf{x}} g(\mathbf{x} \mid \mathbf{x}^{(t)}),
\end{equation}
where $\eta > 0$ is the step size (learning rate). Note that gradient ascent provides an iterative approach that ensures steady improvement toward a locally optimal solution.

For each subproblem, the updates guarantee improvement due to the convexity of the surrogate function constructed via MM, combined with gradient ascent updates for the holographic beamformer. Specifically:
\begin{itemize}
    \item The {generalized eigenvalue solution} ensures that the optimal unconstrained $\mathbf{v}$ is obtained, followed by {power scaling} to enforce feasibility.
    \item The {null-space projection} of $\mathbf{z}$ onto Bob's effective channel ensures that it does not interfere with Bob while maximizing interference at Eve.
    \item The surrogate function for the holographic beamformer satisfies the majorization and the tangent property and the {gradient ascent} update for $\mathbf{w}$ ensures stepwise improvement while keeping the amplitude constraints satisfied through {projection onto the feasible set}.
\end{itemize}

Since the secrecy rate is {upper-bounded} due to the {finite power budget} $ P_t $, and the function is {monotonically increasing}, the sequence of objective function values converges to a {stationary point}. While {global optimality} is not guaranteed due to the non-convex nature of the problem, the {gradient ascent-based MM framework ensures convergence to a locally optimal solution} where no further improvement is possible.

 \subsection{Complexity Analysis}

The computational complexity of the proposed MM-based framework is analyzed by evaluating the three main components: digital beamforming, artificial noise design, and holographic beamforming.

The digital beamforming step involves solving a generalized eigenvalue problem for matrices of size \(R \times R\), where \(R\) is the number of RF chains. The complexity of eigenvalue decomposition is \(\mathcal{O}(R^3)\). Since \(R \ll M\) (the number of RHS elements), this step is computationally efficient. The artificial noise design requires two operations. First, the null space of the legitimate user's effective channel is computed using singular value decomposition (SVD), with a complexity of \(\mathcal{O}(M R^2)\) for a matrix of size \(M \times R\). Second, the alignment of the artificial noise vector with the principal eigenvector of Eve's interference covariance matrix requires eigenvalue decomposition, with a complexity of \(\mathcal{O}(R^3)\). The overall complexity of this step is dominated by the SVD, resulting in \(\mathcal{O}(M R^2)\). The holographic beamforming step uses an iterative gradient ascent method within the MM-based framework. Each iteration involves:
\begin{enumerate}
    \item Computing the gradient, with a complexity of \(\mathcal{O}(M^2)\).
    \item Projecting \(\mathbf{w}\) onto the box-constrained domain \([0, 1]\), with a complexity of \(\mathcal{O}(M)\).
\end{enumerate}
For \(T_h\) iterations, the total complexity of each update is \(\mathcal{O}(T_h M^2)\). The joint optimization algorithm iterates over these three steps for \(T\) outer iterations. The overall complexity is $
\mathcal{O}\left(T \cdot (R^3 + M R^2 + T_h M^2)\right)$. In typical scenarios where \(R \ll M\), the terms \(M R^2\) and \(T_h M^2\) dominate, making holographic beamforming the most computationally intensive. The total complexity simplifies to $\mathcal{O}\left(T \cdot T_h M^2\right)$,
indicating that the algorithm scales quadratically with the number of RHS elements.

\section{Simulation Results} \label{sec_4}  
In this section, we present simulation results to evaluate the performance of the proposed algorithm for secrecy rate maximization in RHS-assisted secure holographic communications. The noise power at Bob and Eve is fixed at \(-75\) dBm, and the results are evaluated for varying transmit power levels. The carrier frequency is set to \(30\) GHz, and the element spacing on the RHS is assumed to be \(\lambda/3\). The holographic beamformer employs iterative amplitude optimization with a convergence threshold of \(\epsilon = 10^{-5}\) and a learning rate of \(\eta = 0.01\). Given the free-space propagation vector \(\mathbf{k}_f\) and the RHS propagation vector \(\mathbf{k}_s\), the relationship between them is governed by the relative permittivity of the RHS substrate, \(\epsilon_r\), typically valued at 3, such that $
|\mathbf{k}_s| = \sqrt{\epsilon_r} |\mathbf{k}_f|$.
For this system, we adopt standard values commonly used in the literature: \( |\mathbf{k}_f| = 200\pi \) and \( |\mathbf{k}_s| = 200\sqrt{3}\pi \), as discussed in \cite{deng2022reconfigurable}. The RHS is assumed to be placed at the altitude of $50$ meters, with the surface aligned in the \((x, y)\)-plane. Bob's position is fixed directly in front of the RHS at a distance of \(100\) meters, while Eve's position is randomly distributed within a circle of radius \(5\) meters centered around Bob.
\begin{figure*}
    \centering
 \begin{minipage}{0.48\textwidth}
     \centering
   \includegraphics[width=\linewidth]{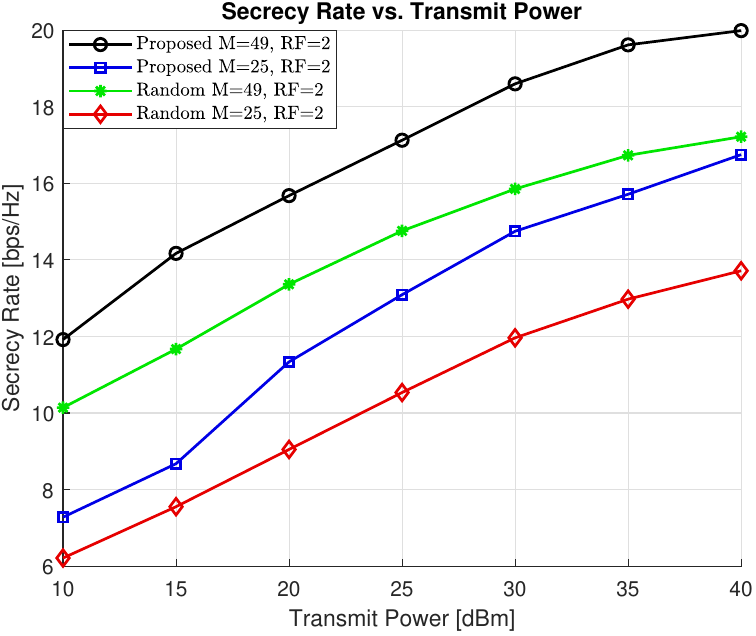}
   \caption{Secrecy rate as a function of the transmit power with varying RHS sizes and $2$ RF chains.}
   \label{fig2}
\end{minipage}  
      \begin{minipage}{0.48\textwidth}
       \centering
   \includegraphics[width=\linewidth]{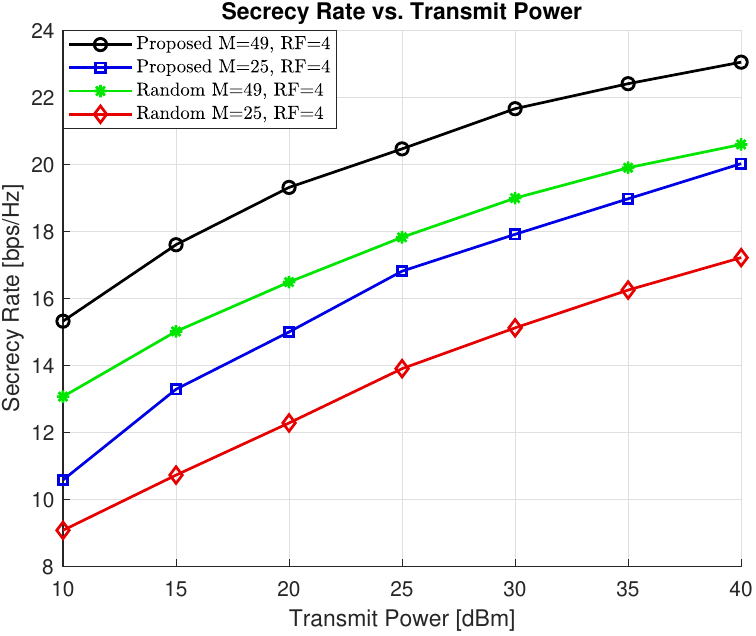}
   \caption{Secrecy rate as a function of the transmit power with varying RHS sizes and $4$ RF chains.}
    \label{fig3}
    \end{minipage}  
\end{figure*} 
We model the channels as a Rician fading given by
\begin{equation}
\mathbf{h} = \sqrt{\frac{K}{K+1}} \mathbf{h}_{\text{LOS}} + \sqrt{\frac{1}{K+1}} \mathbf{h}_{\text{NLOS}},
\end{equation}
where \(\mathbf{h}_{\text{LOS}}\) represents the deterministic LoS component, and \(\mathbf{h}_{\text{NLOS}}\) denotes the non-LoS component. The parameter \(K\) denotes the Rician factor. The LoS component for Bob and Eve is modelled with path loss exponents \(\beta_{b} = 2.2\) and \(\beta_{e} = 2.5\), respectively. The RHS response is be modelled as uniform planner arrays. Unless otherwise stated, the Rician factor is set to \(0\), corresponding to a purely Rayleigh fading scenario. Results with varying Rician factors will also be presented to evaluate performance in LoS-dominant scenarios. The number of RF chains is set to either \(2\) or \(4\), unless otherwise specified.

\begin{figure*}
    \centering
 \begin{minipage}{0.48\textwidth}
     \centering
   \includegraphics[width=\linewidth]{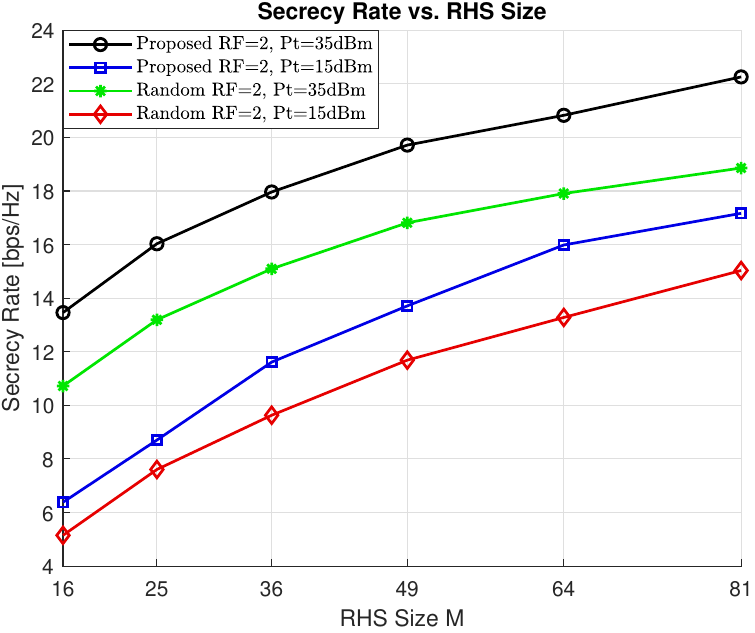}
   \caption{Secrecy rate as a function of the RHS size with varying transmit power and $2$ RF chains.}
   \label{fig4}
\end{minipage}  
      \begin{minipage}{0.48\textwidth}
       \centering
   \includegraphics[width=\linewidth]{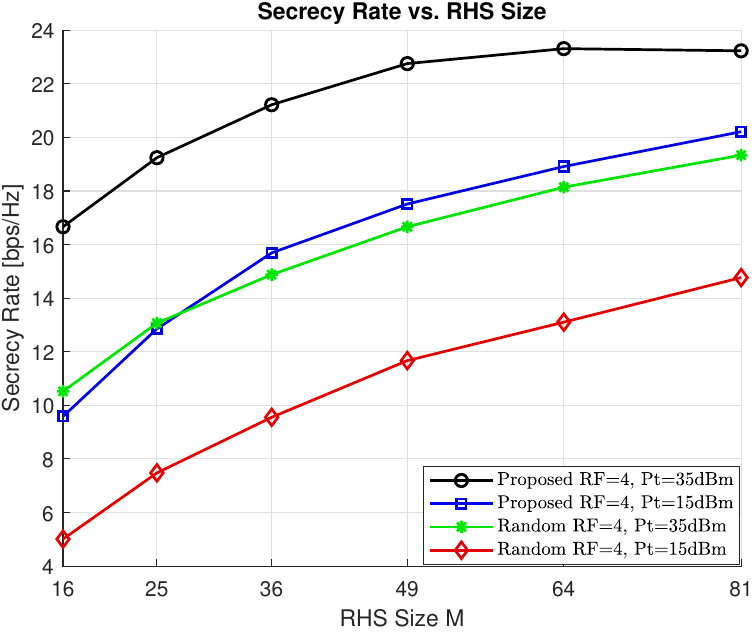}
   \caption{Secrecy rate as a function of the RHS size with varying transmit power and $4$ RF chains.}
    \label{fig5}
    \end{minipage}  
\end{figure*} 
To assess the performance of the proposed method, a benchmark scheme is established. In this scheme, the digital beamformers and AN are assumed to be optimized according to the MM-based framework, while the holographic weights are randomly selected within the range \([0,1]\). Throughout this work, we refer to this scheme as \emph{Random} and to the full MM-based framework developed herein as \emph{Proposed}.

Figure \ref{fig2} illustrates the variation in secrecy rate as a function of the transmit power in dBm, with the number of RF chains fixed at \(R = 2\). Key trends are observed when comparing the proposed and benchmark methods, for RHS sizes \(M = 25\) and \(M = 49\). At low transmit power, the secrecy rate gap between the proposed and random methods is less pronounced for both RHS sizes, as the limited power restricts the overall impact of beamforming optimization. However, even in this region, the proposed method achieves slightly higher secrecy rates due to its more efficient design in suppressing leakage for Eve. As the transmit power increases, the difference between the methods becomes more significant. For \(M = 25\), the gap between the proposed and benchmark methods begins to widen noticeably around $20$ dBm, where the proposed method demonstrates better exploitation of the available degrees of freedom to enhance the signal at Bob and reduce leakage to Eve. In contrast, the benchmark method fails to achieve a comparable improvement, leading to slower growth. For \(M = 49\), the advantage of the proposed method is even more pronounced, as the larger RHS provides additional degrees of freedom for precise beamforming, enabling higher secrecy rates compared to \(M = 25\). At higher power levels, the proposed method for \(M = 49\) outperforms all other configurations, while the benchmark method saturates at much lower secrecy rates. Figure \ref{fig3} illustrates the secrecy rate as a function of the transmit power in dBm, with the number of RF chains set to \(R = 4\). Comparing it with the previous case, note that at low transmit power levels (e.g., $10$–$20$ dBm), the secrecy rate increases modestly for both \(R = 4\) and \(R = 2\) cases, with a relatively small gap between the two configurations. This indicates that at lower power levels, the additional RF chains in \(R = 4\) provide only a marginal improvement in beamforming performance, as the overall system is limited by the low transmit power. Nevertheless, the proposed method consistently outperforms the benchmark scheme, even in this regime. As the transmit power increases beyond 20 dBm, the performance gap between \(R = 4\) and \(R = 2\) widens significantly. For \(R = 4\), the proposed method achieves higher secrecy rates due to the enhanced degrees of freedom provided by the additional RF chains, which improve signal focusing at Bob and signal leakage at Eve. This advantage becomes particularly evident at higher power levels, such as 30 dBm and beyond, where the \(R = 4\) configuration maintains a substantial lead over \(R = 2\). For the benchmark method, the secrecy rate increases much more slowly and saturates earlier for both RF configurations, highlighting the suboptimality of non-optimized beamforming.

\begin{figure*}
    \centering
 \begin{minipage}{0.48\textwidth}
     \centering
   \includegraphics[width=\linewidth]{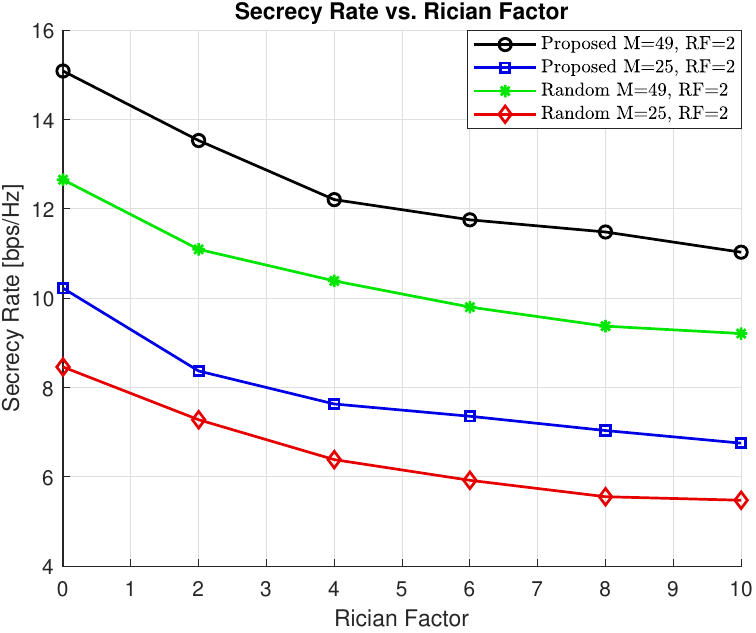}
   \caption{Secrecy rate as a function of the Rician factor with varying transmit RHS size and $2$ RF chains.}
   \label{fig6}
\end{minipage}  
      \begin{minipage}{0.48\textwidth}
       \centering
   \includegraphics[width=\linewidth]{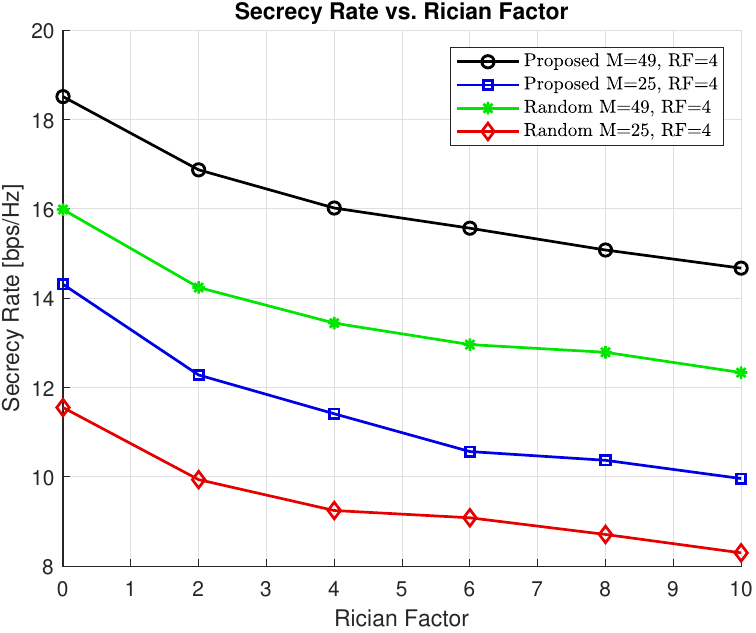}
   \caption{Secrecy rate as a function of the Rician factor with varying transmit RHS size and $4$ RF chains.}
    \label{fig7}
    \end{minipage}  
\end{figure*} 

Figure \ref{fig4} illustrates the variation of secrecy rate as a function  the of RHS size \(M\), with the number of RF chains fixed at \(R = 2\). We can see that the rate consistently increases as \(M\) increases, for both the proposed and benchmark schemes at different transmit power levels, as a larger RHS provides greater flexibility for holographic beamforming. We see that the proposed method outperforms the benchmark scheme across all values of \(M\) at different power levels, with the performance gap becoming more pronounced as \(M\) grows.  At smaller RHS sizes \(M = 16\), the difference between the proposed and benchmark methods at the same power level is less significant due to limited flexibility. However, as \(M\) increases, the proposed method capitalizes on the enhanced spatial capabilities to achieve higher secrecy rates. Despite this improvement, a saturation trend is observed at larger \(M\), where the rate of secrecy rate improvement diminishes. 

Figure \ref{fig5} shows the performance as a function of RHS size \(M\) with the number of RF chains set to \(R = 4\), offering a comparison to the previous case where \(R = 2\). Overall, the secrecy rate is significantly higher for \(R = 4\) across all values of \(M\), reflecting the enhanced spatial degrees of freedom provided by the additional RF chains, which enable more effective beamforming and interference suppression. The proposed method continues to outperform the benchmark scheme, and the performance gap becomes more pronounced for larger \(M\) compared to the \(R = 2\) case. This demonstrates the ability of the proposed method to better utilize the increased RF chains for joint optimization. At smaller \(M\), the secrecy rate differences between \(R = 4\) and \(R = 2\) are less pronounced, as the limited RHS size constrains the spatial resources available for optimization. However, as \(M\) increases, the secrecy rate growth is steeper for \(R = 4\), highlighting the compounding benefit of additional RF chains and larger RHS sizes. Despite this, a saturation trend is observed at higher values of \(M\) for \(R = 4\), though the saturation point occurs at a much higher secrecy rate compared to \(R = 2\).

Figure \ref{fig6} shows the variation of secrecy rate as a function of the Rician factor, with the number of RF chains fixed at $R=2$. The trends demonstrate how the presence of a stronger LoS component, represented by higher Rician factors, influences the system's secrecy performance.
The figure shows that the secrecy rate decreases consistently as the Rician factor increases, indicating that a stronger LoS component negatively impacts the system's ability to maintain secure communication. This trend suggests that the presence of a dominant LoS path reduces the effectiveness of leakage suppression at Eve, due to the reduced randomness and spatial diversity in the channel.  Figure \ref{fig7} shows the performance as a function of the Rician factor with the number of RF chains set to \(R = 4\). Compared to the previous case, the secrecy rate is consistently higher across all values of the Rician factor. This improvement reflects the additional spatial degrees of freedom provided by the increased RF chains, enabling more precise beamforming at Bob and enhanced leakage suppression at Eve. Similar to the \(R = 2\) case, the secrecy rate decreases as the Rician factor increases due to the reduced randomness and spatial diversity in the channel as the LoS component becomes more dominant. It is clear that at higher Rician factors, the secrecy rate for \(R = 4\) is greater than the case of \(R = 2\), highlighting the effectiveness of additional RF chains in maintaining secure communication in LoS-dominant scenarios. 
 \begin{figure}
     \centering
    \includegraphics[width=\linewidth]{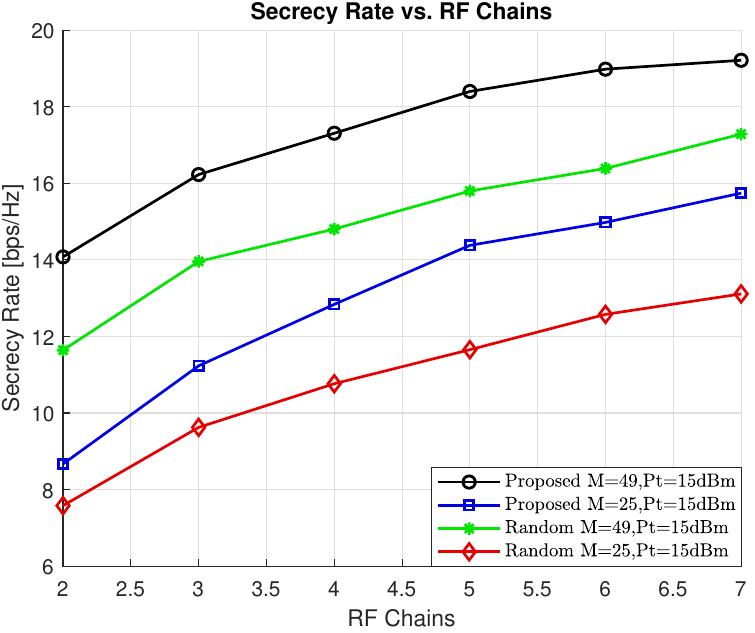}
    \caption{Secrecy rate as a function of the RF chains with varying transmit RHS size and transmit power $15$dBm.}
    \label{fig8}
\end{figure}

Finally, figure \ref{fig8} illustrates the performance as a function of the number of RF chains at a fixed transmit power of $15\sim$dBm for different RHS sizes. The curves reveal that larger RHS sizes consistently achieve higher secrecy rates across all RF chain configurations. This improvement stems from the increased flexibility provided by a larger RHS, which enhances passive beamforming and leakage suppression capabilities. Additionally, the performance gap between different RHS sizes becomes more pronounced as the number of RF chains increases, highlighting the synergistic effect of combining larger RHS sizes with active RF chains. However, a saturation trend is observed at higher RF chain configurations, where the secrecy rate improves at a slower rate, indicating diminishing returns from increasing RF chains.

From the same figure, comparing it with Figures \ref{fig6} and \ref{fig7}, we see that increasing the number of RF chains, has a similar improvement in terms of power as increasing the size of the RHS. While increasing RF chains provides direct performance gains through active beamforming, it also incurs higher hardware complexity, energy consumption, and implementation costs. In contrast, adding more RHS elements to enhance holographic capability occurs at minimal cost and greater energy efficiency, thus making it a highly scalable solution.

In the above figures, we learned how transmit power, RHS size, RF chains, and Rician factor impact the secrecy rate in RHS-enabled systems. At low transmit power, secrecy rate differences between the proposed and benchmark methods are small, but the proposed method consistently outperforms due to better leakage suppression. As power increases, the performance gap widens, particularly for larger RHS sizes, where additional spatial degrees of freedom enable more effective beamforming. 
Increasing the RHS size consistently enhances secrecy rates for both methods, with the proposed approach leveraging the additional flexibility more effectively. Similarly, adding RF chains improves performance by increasing beamforming flexibility, with the proposed method capitalizing on these enhancements to focus signals at legitimate receivers and suppress eavesdropping. However, from comparison, it is clear that enhancing the RHS size is far more desirable than increasing RF chains due to its significantly lower cost and energy requirements. Expanding the RHS size offers a scalable and cost-effective solution, improving holographic beamforming capabilities with minimal additional expense, while adding RF chains increases hardware complexity, energy consumption, and implementation costs. These factors make RHS enhancements a more practical and efficient pathway for achieving high secrecy rates in secure communication systems. We can also conclude that the Rician factor negatively impacts secrecy rates, as stronger LoS components reduce randomness and spatial diversity. Nonetheless, the proposed method mitigates this with advanced beamforming, maintaining higher rates even in LoS-dominant scenarios. 

\section{Conclusion} \label{sec_5}
In this paper, we proposed a novel optimization framework for maximizing the secrecy rate in RHS-assisted communication systems. By jointly optimizing digital beamforming, artificial noise generation, and analog holographic beamforming, the proposed method effectively enhances secure communication against eavesdropping. Leveraging the MM-based framework, the approach systematically addressed the non-convexity of the problem, ensuring reliable convergence to locally optimal solutions. Simulation results demonstrated significant improvements in secrecy rates compared to benchmark schemes, highlighting the advantages of the proposed method in exploiting the spatial degrees of freedom provided by larger RHS sizes and additional RF chains and its robustness at low transmission power and LoS dominant channels. Importantly, the study revealed that enhancing RHS size offers a more cost-effective and energy-efficient means of improving performance than increasing RF chains, making it a highly scalable solution for secure and efficient next-generation wireless networks.

\begin{appendices}
    \section{Proof of Theorem 1}
    \textbf{Proof:} Let $ f(\mathbf{v}) = \log_2\left(1 + \frac{\mathbf{v}^H \mathbf{W}^H \mathbf{H} \mathbf{W} \mathbf{v}}{b}\right) $, where $\mathbf{H} \succeq 0$ (positive semidefinite), $\mathbf{W} \in \mathbb{C}^{M \times F}$, and $b > 0$. To prove the theorem, we need to majorize $f(\mathbf{v})$ with an upper-bound function. Using the definition of $f(\mathbf{v})$, we can write it as
\begin{equation}
f(\mathbf{v}) = \frac{\ln\left(1 + \frac{\mathbf{v}^H \mathbf{W}^H \mathbf{H} \mathbf{W} \mathbf{v}}{b}\right)}{\ln(2)}.
\end{equation}

Let $x = \frac{\mathbf{v}^H \mathbf{W}^H \mathbf{H} \mathbf{W} \mathbf{v}}{b}.$ Thus, the function becomes $f(\mathbf{v}) = \frac{\ln(1 + x)}{\ln(2)}$. The logarithm function $\ln(1 + x)$ is concave. For concave functions, a first-order Taylor expansion provides a global upper bound. Let $x^{(t)}$ represent the value of $x$ at iteration $t$. The first-order Taylor expansion is given by:
\begin{equation}
\ln(1 + x) \leq \ln(1 + x^{(t)}) + \frac{x - x^{(t)}}{1 + x^{(t)}}.
\end{equation}

Substituting $x = \frac{\mathbf{v}^H \mathbf{W}^H \mathbf{H} \mathbf{W} \mathbf{v}}{b}$, the bound becomes:
 \begin{equation}
\begin{aligned}
    \ln\left(1 + \frac{\mathbf{v}^H \mathbf{W}^H \mathbf{H} \mathbf{W} \mathbf{v}}{b}\right) 
    &\leq \ln\left(1 + \frac{\mathbf{v}^{(t)H} \mathbf{W}^H \mathbf{H} \mathbf{W} \mathbf{v}^{(t)}}{b}\right) \\
    &\quad + \frac{\frac{\mathbf{v}^H \mathbf{W}^H \mathbf{H} \mathbf{W} \mathbf{v}}{b} - \frac{\mathbf{v}^{(t)H} \mathbf{W}^H \mathbf{H} \mathbf{W} \mathbf{v}^{(t)}}{b}}{1 + \frac{\mathbf{v}^{(t)H} \mathbf{W}^H \mathbf{H} \mathbf{W} \mathbf{v}^{(t)}}{b}}.
\end{aligned}
\end{equation}

By rewriting this, we get 
 \begin{equation}
\begin{aligned}
    \ln\left(1 + \frac{\mathbf{v}^H \mathbf{W}^H \mathbf{H} \mathbf{W} \mathbf{v}}{b}\right) 
    &\leq \ln\left(1 + \frac{\mathbf{v}^{(t)H} \mathbf{W}^H \mathbf{H} \mathbf{W} \mathbf{v}^{(t)}}{b}\right) \\
    &\quad + \frac{\frac{\mathbf{v}^H \mathbf{W}^H \mathbf{H} \mathbf{W} \mathbf{v}}{b}}{1 + \frac{\mathbf{v}^{(t)H} \mathbf{W}^H \mathbf{H} \mathbf{W} \mathbf{v}^{(t)}}{b}} \\
    &\quad - \frac{\frac{\mathbf{v}^{(t)H} \mathbf{W}^H \mathbf{H} \mathbf{W} \mathbf{v}^{(t)}}{b}}{1 + \frac{\mathbf{v}^{(t)H} \mathbf{W}^H \mathbf{H} \mathbf{W} \mathbf{v}^{(t)}}{b}}.
\end{aligned}
\end{equation}

The last term cancels out the constant.

The bound for $\ln(1 + x)$ is now scaled by $\frac{1}{\ln(2)}$ to revert to $f(\mathbf{v})$
\begin{equation}
f(\mathbf{v}) \leq f(\mathbf{v}^{(t)}) + \frac{\frac{\mathbf{v}^{(t)H} \mathbf{W}^H \mathbf{H} \mathbf{W} \mathbf{v}^{(t)}}{\ln(2) b}}{1 + \frac{\mathbf{v}^{(t)H} \mathbf{W}^H \mathbf{H} \mathbf{W} \mathbf{v}^{(t)}}{b}} \cdot \mathbf{v}^H \mathbf{W}^H \mathbf{H} \mathbf{W} \mathbf{v}.
\end{equation}

By simplifying the coefficient $\mathbf{v}^H \mathbf{W}^H \mathbf{H} \mathbf{W} \mathbf{v}$ yields
\begin{equation}
f(\mathbf{v}) \leq f(\mathbf{v}^{(t)}) + \frac{\mathbf{v}^{(t)H} \mathbf{W}^H \mathbf{H} \mathbf{W} \mathbf{v}^{(t)}}{\ln(2) \left(b + \mathbf{v}^{(t)H} \mathbf{W}^H \mathbf{H} \mathbf{W} \mathbf{v}^{(t)}\right)} \mathbf{v}^H \mathbf{W}^H \mathbf{H} \mathbf{W} \mathbf{v}.
\end{equation}

Thus, the theorem is proven. $\Box$

\end{appendices}

\ifCLASSOPTIONcaptionsoff
  \newpage
\fi

{\footnotesize
\bibliographystyle{IEEEtran}
\bibliography{main}}
  
\end{document}